\documentclass[oribibl]{llncs}

\usepackage[utf8x]{inputenc}
\usepackage[english]{babel}
\usepackage[T1]{fontenc}

\usepackage{amsmath}
\usepackage{graphicx}
\usepackage{url,hyperref}

\usepackage{amssymb}

\usepackage{enumerate}

\newcommand{\tabincell}[2]{\begin{tabular}{@{}#1@{}}#2\end{tabular}}

\newcommand{\cD}{\mathcal{D}} 
\newcommand{\cX}{\mathcal{X}} 
 
\newcommand{\cP}{\mathcal{P}} 
\newcommand{\Hm}[1]{\mathbf{H}_{\infty}\left(#1\right)}
\newcommand{\HmAv}[1]{\widetilde{\mathbf{H}}_{\infty}\left(#1\right)}
\newcommand{\Hmtr}[3]{\mathbf{H}^{\textup{Metric},#2,#3}_{}\left(#1\right)}
\newcommand{\HmtrAv}[3]{\widetilde{\mathbf{H}}^{\textup{Metric},#2,#3}_{}\left(#1\right)}
\newcommand{\HhllAv}[3]{\widetilde{\mathbf{H}}^{\textup{HILL},#2,#3}_{}\left(#1\right)}
\newcommand{\Hhll}[3]{\mathbf{H}^{\textup{HILL},#2,#3}_{}\left(#1\right)}

\newcommand{\HmtrdAv}[3]{\widetilde{\mathbf{H}}^{\textup{Metric-d},#2,#3}_{}\left(#1\right)}
\newcommand{\Habs}[3]{\mathbf{H}^{\textup{|Metric|},#2,#3}_{}\left(#1\right)}
\newcommand{\HabsAv}[3]{\widetilde{\mathbf{H}}^{\textup{|Metric|},#2,#3}_{}\left(#1\right)}
\newcommand{\Ev}[3]{\mathbf{E}_{#3 \leftarrow #2} #1\left( #3 \right)}
\newcommand{\EvAdd}[4]{\mathbf{E}_{#3 \leftarrow #2} #1\left( #3, #4 \right)}
\newcommand{\Max}{\mathrm{Max}}
\newcommand{\NP}{\mathbf{NP}}

\newcommand{\bigO}[1]{\mathcal{O}\left( #1 \right)}

\newcommand{\eqd}{\,{\buildrel d \over =}\,}

\begin{document}

\title{Modulus Computational Entropy \thanks{This work was partly supported by the WELCOME/2010-4/2 grant founded
within the framework of the EU Innovative Economy Operational
Programme.
} }
\author{Maciej Sk\'{o}rski \thanks{Cryptology and Data Security Group, University of Warsaw. Email: \email{maciej.skorski@gmail.com}}}
\institute{University of Warsaw}

\maketitle

\begin{abstract}
The so-called {\em leakage-chain rule} is a very important tool used in many security proofs. It gives an upper bound on the entropy loss of a random variable $X$ in case the adversary who having already learned some random variables $Z_{1},\ldots,Z_{\ell}$ correlated with $X$, obtains some further information $Z_{\ell+1}$ about $X$. Analogously to the information-theoretic case, one might expect that also for the \emph{computational} variants of entropy the loss depends only on the actual leakage, i.e. on $Z_{\ell+1}$. Surprisingly, Krenn et al.\ have shown recently that for the most commonly used definitions of computational entropy this holds only if the computational quality of the  entropy deteriorates exponentially in $|(Z_{1},\ldots,Z_{\ell})|$. This means that the current standard definitions of computational entropy do not allow to fully capture leakage that occurred ''in the past", which severely limits the applicability of this notion. 

As a remedy for this problem we propose a slightly stronger definition of the computational entropy, which we call the  \emph{modulus computational entropy}, and use it as a technical tool that allows us to prove a desired chain rule that depends only on the actual leakage and not on its history. Moreover, we show that the  modulus computational entropy unifies other,sometimes seemingly unrelated, notions already studied in the literature in the context of information leakage and chain rules. Our results indicate that the modulus entropy is, up to now, the weakest restriction that guarantees that the chain rule for the computational entropy works. As an example of application we demonstrate a few  interesting cases where our restricted definition is fulfilled and the chain rule holds.
\end{abstract}

\pagestyle{headings}
\setcounter{page}{1}

\section{Introduction}
Entropy is the most fundamental concept in Information Theory. First introduced in this context by Shannon \cite{Shannon1948}, as a measure of the uncertainty associated with a probability distribution, it has been generalized in many ways. The commonly used generalization of Shannon Entropy is R\'{e}nyi Entropy, defined for any arbitrary nonnegative order, which includes Shannon Entropy as a special case of order 1. Informally, a reasonable entropy measure indicates for a given distribution how much randomness it contains. According to this intuition, distributions uniform over large sets should have very high entropy, in opposite to distributions which has small support or hit a small set with high probability, being easy to predict. 

\paragraph{Indistinguishability and entropy.} The notion of entropy has been generalized also for the purpose of Computational Complexity Theory and Cryptography, to take \emph{computational} aspects into account. The reader might wish to refer to \cite{Reyzin2011} for a short survey.\ Historically computational entropy was first introduced in \cite{Yao1982} and, basing on a different concept, in \cite{HILL99}. This last approach, based on the notion of \emph{indistinguishability}, is the one we follow in this work. Let us try to give some intuitions here (the precisely definitions will be given in Section 2).\ To define computational entropy of $X$, one relaxes the requirement that $X$ should have entropy itself. Instead, we assume that $X$ is only close to a distribution $Y$ which has suitable information-theoretic entropy. We have to specify two things: (a) the notion of entropy we use and (b) what does it mean "being close". 
To give a rigorous formulation of (b), one uses the  concept of \emph{distinguishing}, borrowed from convex analysis and topology. Function $D$ separates (\emph{distinguishes}) a set $\mathbb{X}$ from another set $\mathbb{Y}$ with \emph{advantage} at least $\epsilon$ if $D(x) - D(y) \geqslant \epsilon$ for every $x\in\mathbb{X},\,y\in\mathbb{Y}$. In turn, for a predefined class $\cD$ of functions, two sets are said to be $\left(\cD,\epsilon\right)$-indistinguishable, if there is no $D\in\cD$ that can distinguish between these two sets with advantage greater than $\epsilon$. The smaller $\epsilon$ and the wider class $\cD$ we take, the stronger indistingusihability we obtain. Especially, 
indistinguishability applied to two probability 
distributions (as one-element sets) and all boolean functions (as distinguishers), where acting $D$ on a  distribution $\mathbf{P}_{X}$ is defined by $D\left(\mathbf{P}_{X}\right) = \mathbf{E}_{x\leftarrow X}D(x)$, yields the definition of the statistical distance. In applications involving computational complexity, one usually use circuits of bounded size as a class of distinguishers.

\paragraph{Leakage Lemma and Chain Rule.}
Leakage lemma is the term commonly used in referring to various generalizations of the observation which, saying less formally, states that min-entropy of a distribution $X$ conditioned on another distribution $Z$ distributed over $\{0,1\}^{m}$ decreases, with respect to min-entropy of $X$, by at most $m$ (the number of bits in the string encoding $Z$). 
The name comes from security-related applications, where one considers entropy of a distribution conditioned on information that might have been revealed to the adversary. The larger difference between entropy of a distribution and entropy of the corresponding conditioned distribution, the larger leakage is; such an approach, based on computational entropy, was used first by Dziembowski and Pietrzak \cite{Dziembowski2008}. In turn, the term leakage chain rule is used to state the same principle for the case when we are given entropy of $X$ conditioned on $Z_1$, and observing some further leakage $Z_2$ ask for the entropy of $X$ conditioned on $Z_1 Z_2$. Such conditioning of an already conditioned distribution refers to the so called "leakage-after-leakage" scenario. The name ``Leakage Chain Rule'' comes from the fact that we think of $Z_1$ and $Z_2$ as information about $X$ that ``leaked'' subsequently to the adversary.

For commonly used information-theoretic notions of conditional entropy, the chain rule is known to be true, i.e. the loss in entropy depends on $\left|Z_2\right|$ and not on $Z_1$. The problems appear in the case of computational generalizations of entropy. The computational leakage lemma \cite{Dziembowski2008,Fuller2011}, turned out not to give rise naturally to the leakage chain rule at least for important indistinguishability based definitions of conditional computational entropy and was addressed as an open problem \cite{FullerReyzin2012}. The computational leakage chain rule was proved only for specific scenarios, either by adding strong assumptions to definitions \cite{Fuller2011, Chung2011}, or by changing definitions (see \cite{Reyzin2011} for the discussion of computational relaxed entropy based on Leakage Lemma \cite{GentryWichs2010}). Recently, a counterexample to the chain rule for computational min entropy has been found \cite{Pietrzak2013}. It shows that the computational entropy of $X|Z_1 Z_2$ can decrease dramaticaly with respect to the entropy of $X|Z_1$, even if $Z_2$ is just a one bit.
\paragraph{Our contribution.}
Interested in establishing the (possibly) weakest condition to make the leakage chain rule work for the 'standard'  computational entropy (i.e. defined using indistinguishablity and the min-entropy), we define the modulus computational entropy and show that its definition is satisfied by technical assumptions which have been used by other authors to obtain a chain rule: the decomposable entropy introduced by Fuller and Reyzin \cite{Fuller2011} and the samplability assumption used by Chung et al.~in \cite{Chung2011}. Interestingly, it is implied by  the ''squared-indistinguishability" introduced in \cite{Dodis2013}. Furthermore, we investigate three cases that has not been considered yet: (a) when computational entropy is almost maximal, (b) the existence of an $\NP$ oracle over the domain of $X$ to which distinguishers are given access \footnote{We stress that this is a non-trivial result, as the computational entropy $X$ given $Z$ is calculated by distiguishers on $\{0,1\}^{n+m}$, thus it might happen that even circuits of size $2^{n}$ are not able to break it.}, and (c) when the leakage is relatively short. In all these cases our definition is fulfilled and the chain rule works. Summing up, we reduce already known necessary conditions to the one simpler concept and show a few new non-trivial cases where the chain rule works.
\paragraph{Our techniques}  
We observe that to ensure the chain rule, one need to control \emph{the conditional advantages of distinguishers}, i.e. advantages calculated conditionally on appropriate events. The same concept appears in \cite{Dodis2013}.
This elementary technique leads to quite non-trivial results and we believe that its application can be of independent interest.
\paragraph{Outline of the work.} 
Section 2 deals with some preliminary concepts, conventions and notations. In Section 3 we explain basic definitions and terminology related to the computational entropy, and discuss the positive and negative results related to the leakage chain rule problem. In Section 4 we introduce our main tool-- the modulus entropy and show that  the leakage chain rule holds for this notion. Section 5 contains a brief summary of the most important consequences of our results - estimating the cost of conversions to the modulus entropy from several technical assumptions. Section 6 provides their proofs. 

\section{Preliminaries} \label{Preliminaries}

Throughout this work we assume that all random variables are defined on some finite probability space and they take values in $\{0,1\}^{*}$. If $X$ is a random variable then $\mathbf{P}_{X}$ will be its distribution.  
Writing $X\in S$ we mean that $X$ takes its values in the set $S$. By $|S|$ we denote the cardinality of $S$. For two random variables $X,Z$ by $X|Z=z$ we denote the distribution of $X$ conditioned on $Z=z$ and $(X,Z)$ means the concatenation of $X$ and $Z$. For every $n$, by $U_{n}$ we denote the uniform distribution over $\{0,1\}^{n}$. By $\left(\mathrm{det}\{0,1\},s\right)$ and $\left(\mathrm{\det}[0,1],s\right)$ we mean the class of all deterministic circuits of size at most $s$, with output in the set $\{0,1\}$ and $[0,1]$ respectively. Similarly, we denote by $\left(\mathrm{rand}\{0,1\},s\right)$ the set of all randomized boolean circuits of size at most $s$. All logarithms are taken 
to the base $2$. For $D : \cX \rightarrow [0,1]$ and $k \leq \log|\cX|$ we denote by $\Max^k_{D} \subseteq \cX$ a set of cardinality $2^k$ such that for every $x \in \Max^k_{D}$ and every $x' \not\in \Max^k_{D}$ we have $D(x) \geqslant D(x')$. For a boolean function $D$, we write $|D| = \sum_{x\in X}D(x)$.

\subsection{Min Entropy}
We start with recalling information-theoretic notions.
\begin{definition}[Min Entropy]
A random variable $X$ has at least $k$ bits of min-entropy, denoted by $\Hm{X}\geqslant k$, if and only if $\max_{x}\mathbf{P}_X\left(x\right)\leqslant 2^{-k}$. 
\end{definition}
\noindent The \emph{conditional} min-entropy can be defined in two ways, both compatible with the above definition.
The first one is given below. 

\begin{definition}[Worst-Case Conditional Min-Entropy]
Given a pair of random variables $(X,Z)$ we say that\emph{ $X$ conditioned on $Z$ has the min-entropy at least $k$} and denote it by $\mathbf{H}_{\infty}(X|Z) \geqslant k$, if and only if for every $z$ we have 
\begin{displaymath} \max\limits_{x}\mathbf{P}_{X|Z=z}\left(x\right)\leqslant 2^{-k}. \end{displaymath}
\end{definition}

\noindent
It is called the \emph{worst-case} because it requires $X$ to have high min-entropy when it is conditioned on the event ``$Z=z$''  for {\em every} $z$.  The alternative definition requires this fact to hold {\em on average}:
\begin{definition}[Average Conditional Min-Entropy \cite{Dodis2008}]
Given a pair of random variables $(X,Z)$ we say that \emph{$X$ conditioned on $Z$ has the average min-entropy at least $k$} and denote $\widetilde{\mathbf{H}}_{\infty}\left(X|Z\right) \geqslant k$,  if and only if \begin{displaymath}\mathbf{E}_{z\leftarrow Z}\left[ \max\limits_{x}\mathbf{P}_{X|Z=z}\left(x\right) \right]\leqslant 2^{-k}. \end{displaymath} 
\end{definition}
\noindent Usually it is not so important which one of these definitions is used, as one can convert the average conditional min entropy into the worst-case variant.
\begin{lemma}[See \cite{Dodis2008}, Lemma 2.2]\label{Average_WorstCase_Conversion}
Suppose that $\widetilde{\mathbf{H}}_{\infty}\left(X|Z\right) \geqslant k$. Then holds $\mathbf{H}_{\infty}\left(X|Z=z\right) \geqslant k-\log\frac{1}{\delta}$ with probability at least $1-\delta$ over $z\leftarrow Z$.  
\end{lemma}

\subsection{Indistinguishability}\label{sec:indist}
\noindent Below we outline the concept of \emph{indistinguishability}, being a key point in defining computational entropy.
\begin{definition}\label{Distinguishability_Def}
Let $\mathbb{X}$ and $\mathbb{Y}$ be subsets of some set $\mathcal{P}$. Given $\epsilon >0$ we say that a function $F:\,\mathcal{P}\rightarrow [0,1]$ \emph{distinguishes between $\mathbb{X}$ and $\mathbb{Y}$ with advantage at least $\epsilon$} if for every $x \in \mathbb{X}$ and every $y\in\mathbb{Y}$ we have 
$ \left| F(x) - F(y) \right| \geqslant \epsilon. $
\end{definition}
\begin{definition}\label{Indistinguishability_Def}
Let $\mathbb{X}$ and $\mathbb{Y}$ be as in Def. \ref{Distinguishability_Def}.\ Given a class $\mathcal{F}$ of $[0,1]$-valued functions on $\mathcal{P}$, we say that \emph{ $\mathbb{X}$ and $\mathbb{Y}$ are $\left(\mathcal{F},\epsilon\right)$-indistinguishable} if there is no $F\in\mathcal{F}$ that can distinguish between $\mathbb{X}$ and $\mathbb{Y}$ with advantage greater than $\epsilon$.
\end{definition}
\noindent 
In this paper we are mostly interested in the case when $\cP$ is equal to the set of all probability distributions over some finite space $\Omega$.  In this case, every function $D : \Omega \rightarrow [0,1]$ gives rise to a distinguisher $F_D : \cP \rightarrow [0,1]$ defined as $F_D(\mu) = \mathbf{E}_{x\leftarrow \mu} D(x)$.
Thus, we will overload the notation and say that \emph{$D$ distinguishes between $\mathbb{X}$ and $\mathbb{Y}$ with advantage at least $\epsilon$} if the corresponding function $F_D$ distinguishes between $\mathbb{X}$ and $\mathbb{Y}$ with advantage at least $\epsilon$.   We note that $D$ can also be a \emph{randomized} function, which receives an additional input $R$ chosen independently at random. The  expectation $\mathbf{E}D(\cdot)$ is then taken also over $R$. 

\section{Computational Entropy and Leakage - previous works}
As mentioned before, computational entropy can be obtained by generalizing entropy notions in many ways. We follow the approach based on indistinguishability as it seems to be the most standard way and was originally used for studying leakage \cite{Dziembowski2008} as well as in further leakage-related works \cite{Chung2011,Fuller2011,GentryWichs2010}. 
\subsection{Defining Computational Entropy}
\paragraph{Three-layer definition.}
There are three key points, essential for defining computational entropy via indistinguishability: 
\begin{enumerate}[{} (a)]
 \item \label{AbstractEntropy:a} specify, for every $k$, what it means that a distribution ``has (non-computational) entropy at least $k$'',
 \item \label{AbstractEntropy:b} model the adversary, in particular define his computational power, and determine his maximal acceptable success probability, and
 \item \label{AbstractEntropy:c} define the measure of the ``computational distance'' between a given distribution and the set of distributions with entropy at least $k$ (in the sense of (\ref{AbstractEntropy:a})).
\end{enumerate}
In (\ref{AbstractEntropy:a}) one usually uses information-theoretic notion of entropy, most often the min-entropy\footnote{We use only min-entropy in this work. See, however, \cite{Vadhan2012} for a similar definition based on Shanon Entropy.}. For (\ref{AbstractEntropy:b}) one uses a pair $\left(\mathcal{D},\epsilon\right)$ within the framework described in Section  \ref{sec:indist}. Finally, a rigorous formulation of (\ref{AbstractEntropy:c}) can be given in two ways, traditionally called the ``HILL'' or the ``Metric'' version. In the HILL version, defining entropy of a random variable $X$, we require $X$ to be indistinguishable from \emph{one single} distribution with high entropy (in the sense of (\ref{AbstractEntropy:a})), whereas in the definition of the Metric Entropy we require $X$ to be indistinguishable from the set of \emph{all} of high-entropy distributions, which is a bit \emph{weaker} assumption.  The formal definitions below are provided for the conditional versions of both notions. The 
unconditional versions, denoted $\Hhll{X}{\mathcal{D}}{\epsilon}$  and $\Hmtr{X}{\mathcal{D}}{\epsilon}$, are special cases obtained by fixing in the definitions below $Z$ to be constant.

\begin{definition}[HILL Computational Conditional Entropy \cite{Reyzin2007}]\label{HILLEntropy_Definition}
Let $X,Z$ be random variables taking values in $\{0,1\}^{n}$ and $\{0,1\}^{m}$ respectively.
Given $\epsilon > 0$ and a class of distinguishers $\cD$, we say that \emph{$X$ conditioned on $Z$ has at least $k$ bits of computational HILL entropy against $\left(\cD,\epsilon\right)$} and denote by $\HhllAv{X|Z}{\mathcal{D}}{\epsilon}\geqslant k$, if \emph{there exists} a random variable $Y\in\{0,1\}^{n}$ satisfying $\HmAv{Y|Z}\geqslant k$, such that $(X,Z)$ is $\left(\mathcal{D},\epsilon\right)$-indistinguishable from $(Y,Z)$ .
\end{definition}

\begin{definition}[Metric Computational Conditional Entropy \cite{Reyzin2007}]\label{MetricEntropy_Definition}
With $\epsilon,\cD,X$ and $Z$ as in Def.\ \ref{HILLEntropy_Definition}, we say that \emph{$X$ conditioned on $Z$ has at least $k$ bits of computational metric entropy against $\left( \cD,\epsilon\right)$}, denoting $\HmtrAv{X|Z}{\cD}{\epsilon}\geqslant k$, if $(X,Z)$ is $\left(\mathcal{D},\epsilon\right)$-indistinguishable from the set of \emph{all} distributions $(Y,Z)$, satisfying $\HmAv{Y|Z}\geqslant k$.
\end{definition}

\noindent Usually one formulates both definitions more explicity without using the very general notion of indistinguishability (as in Def. \ref{Indistinguishability_Def})
\begin{enumerate}[{} ]
\item[] Definition \ref{HILLEntropy_Definition}: $\HhllAv{X|Z}{\mathcal{D}}{\epsilon}\geqslant k$ if there exists a random variable $Y\in\{0,1\}^{n}$, $\HmAv{Y|Z}\geqslant k$ satisfying
$\left| \mathbf{E} D(X,Z) - \mathbf{E} D(Y,Z) \right| \leqslant \epsilon$ for all $D\in\mathcal{D}$. 
\item[] Definition \ref{MetricEntropy_Definition}: $\HmtrAv{X|Z}{\mathcal{D}}{\epsilon}\geqslant k$ if for every $D\in\mathcal{D}$ there exists a random variable $Y\in\{0,1\}^{n}$, $\HmAv{Y|Z}\geqslant k$ and $\left| \mathbf{E} D(X,Z) - \mathbf{E} D(Y,Z) \right| \leqslant \epsilon$. 
\end{enumerate}
However, our, more general, definitional approach appears to be more useful for the applications presented in the sequel. The definitions of the HILL Computational \emph{Worst}-Case Conditional Entropy $\Hhll{X|Z}{\cD}{\epsilon}$ and the Metric Computational \emph{Worst}-Case Conditional Entropy $\Hmtr{X|Z}{\cD}{\epsilon}$ are obtained by replacing $\HmAv{Y|Z}\geqslant k$ in Def.\ \ref{HILLEntropy_Definition} and Def.\ \ref{MetricEntropy_Definition} (resp.) with $\Hm{Y|Z} \geqslant k$.  



\paragraph{The equivalence between the HILL and Metric-type Entropy.}
It has been observed that the Metric Entropy is more convenient for proving leakage-related results and, in fact,  appears in all such proofs more or less implicitly. Fortunately, there exists a conversion from the Metric Entropy (against real-valued circuits) to HILL Entropy \cite{Barak2003}. This result in its full generality can be stated as follows:
\begin{theorem}[Generalization of \cite{Barak2003}, Thm. 5.2]\label{MetricHILL_Conversion}
Let $\cP$ be the set of all probability measures over $\Omega$. Suppose that we are given a class $\mathcal{D}$ of $[0,1]$-valued functions on $\Omega$, with the following property: if $D\in \cD$ then $D^{c}=^{\textup{def}}\mathbf{1}-D \in \cD$. For $\delta > 0$, let $\mathcal{D'}$ be the class consisting of all convex combinations of length $\mathcal{O}\left(\frac{\log|\Omega|}{\delta^2}\right)$ over $\mathcal{D}$. Let $\mathcal{C}\subset \cP$ be any arbitrary convex subset of probability measures and $X \in \cP$ be a fixed distribution. Consider the following statements:
\begin{enumerate}[{} (a)]
\item \label{itm:HILLMetricCoversion_a} $X$ is $\left(\mathcal{D},\epsilon+\delta\right)$ indistinguishable from \emph{some} distribution $Y\in \mathcal{C}$
 \item \label{itm:HILLMetricCoversion_b} $X$ is $\left(\mathcal{D'},\epsilon\right)$ indistinguishable from the set of \emph{all} distribution $Y\in \mathcal{C}$
\end{enumerate}
Then (\ref{itm:HILLMetricCoversion_b}) implies (\ref{itm:HILLMetricCoversion_a}).  
\end{theorem}
This more general statement follows by inspection of the original proof.
\begin{remark}\label{MetricHILL_Conversion_Discussion}
By choosing $\Omega = \{0,1\}^{n+m}$, a random variable $Z\in\{0,1\}^{m}$ and $\mathcal{C}$ to be the set of all distributions $(Y,Z)$ satisfying $\Hm{Y|Z}\geqslant k$ or alternatively 
$\HmAv{Y|Z}\geqslant k$, we obtain the conversion from Metric Conditional Entropy to HILL Conditional Entropy, for both: the worst and average case variants.
\end{remark}



\subsection{Leakage Rules} 
We are now ready to state the leakage chain rule for conditional min-entropy and compare it with its known generalizations to computational case.  Generally, we are interested in the following problem: 
\begin{quote}
Suppose we have a pair of random variables $(X,Z_1)$ and we know the conditional entropy of $X$ given $Z_1$.  What is the lower bound on the entropy of $X$ given $(Z_1,Z_2)$, where $Z_2$ is some other (possibly correlated) random variable? 
\end{quote}
In the information-theoretic case we have the following result  (cf. \cite{Dodis2008}, Lemma 2.2)
\begin{lemma}[Leakage Chain Rule]\label{ChainRule_MinEntropy}
Let $X,Z_1,Z_2$ be random variables over $\{0,1\}^{n},\{0,1\}^{m_1},\{0,1\}^{m_2}$ respectively. Then
\begin{equation}\label{eq:chain}
\widetilde{\mathbf{H}}_{\infty}\left(X|Z_1,Z_2\right) \geqslant \widetilde{\mathbf{H}}_{\infty}\left(X|Z_1\right)-m_2
\end{equation}
\end{lemma}

\noindent

In the computational framework, the first leakage-related result appeared in \cite{Dziembowski2008} and (formulated in  a different way) in \cite{Reingold2008}. The parameters were improved next in \cite{Fuller2011}. 
\begin{lemma}[Leakage Lemma \cite{Fuller2011}]\label{LeakageLemma_Computational}
Let $X$ and $Z$ be random variables over $\{0,1\}^{n}$ and $\{0,1\}^{m}$, resp. Then 
\begin{equation*}
 \widetilde{\mathbf{H}}^{\textup{Metric},(\mathrm{det}[0,1],s'),\epsilon'}\left(X|Z\right) \geqslant \mathbf{H}^{\textup{Metric},(\mathrm{det}[0,1],s),\epsilon}\left(X\right) - m
\end{equation*}
where $s'=s+\mathcal{O}(1), \ \epsilon' = 2^{m}\epsilon$ and $(\mathrm{det}[0,1],s)$ stands for the class of  circuits (as defined in Section \ref{Preliminaries}).
\end{lemma}
\noindent Let us observe, at least under assumption that there exists an exponentially secure pseudorandom generator, that both losses: in quantity (by $m$ bits) and security measured as $s/\epsilon$ (by factor almost equal to $2^{m}$) can appear simultaneously\footnote{The question whether it can happen was raised in \cite{Fuller2011}} -  
 see Theorem \ref{LeakageLemma_Computational_OptimalLoss} in Appendix \ref{app:LeakageLemma_Thightness}.  

\paragraph{Leakage Chain Rule for Computational Entropy - negative and positive results}
It is a natural question to ask if the Leakage Chain Rule (Lemma \ref{ChainRule_MinEntropy}) can be ``translated'' into the computational version.  In particular, one might be tempted to conjecture that for $X,Z_1$ and $Z_2$ as in Lemma \ref{ChainRule_MinEntropy} it holds that 
\begin{equation}\label{ChainRule_GeneralForm}
\widetilde{\mathbf{H}}^{\textup{Metric},[0,1],s',\epsilon'}\left(X|Z_1,Z_2\right) \geqslant^{\textup{?}} \widetilde{\mathbf{H}}^{\textup{Metric},[0,1],s,\epsilon}\left(X|Z_1\right) - m_2,
\end{equation}
with the security loss of factor $2^{m_2}$ or $\mathrm{poly}\left( 2^{m_2},1/\epsilon\right)$ for the above stated in terms of HILL entropy, where by the security loss we mean $\frac{s'}{\epsilon'} / \frac{s}{\epsilon}$ (reduces to $\epsilon/ \epsilon'$ if $s' \approx s$). Unfortunately, this conjecture is false in general \cite{Pietrzak2013}.  On the positive side, some progress towards proving it for restricted definitions of entropy has been recently made  \cite{Fuller2011, Chung2011, Reyzin2011, GentryWichs2010}. In \cite{Fuller2011}, the authors use an assumption called {\em decomposability}:

\begin{definition}[\cite{Fuller2011}]\label{DecomposableEntropy_Fuller}
Let $X,Z$ be as in Lemma \ref{LeakageLemma_Computational}. Given the parameter $s$, we say that \emph{the decomposable metric-entropy of $X$ conditioned on $Z$ is at least $k$}  and denote $\HmtrdAv{X|Z}{[0,1],s}{\epsilon}\geqslant k$, if for every $z$ 
\[
\HmtrAv{X|Z=z}{[0,1],s}{\epsilon(z)}\geqslant k(z)
\]
where $\epsilon(z), k(z)$ are some numbers satisfying 
$\mathbf{E}\left[ 2^{-k(Z)} \right] = 2^{-k}$ and $\mathbf{E} \left[ \epsilon(Z) \right] \leqslant \epsilon$.
\end{definition}
\noindent
Using this definition they are able to prove the following.
\begin{theorem}[\cite{Fuller2011}] Let $X,Z_1,Z_2$ be as in Lemma  \ref{ChainRule_MinEntropy}\label{ChainRule_ifDecomposable}. Then
for $s'\approx s,$ and $\epsilon' = 2^{m_2}\epsilon$, we have 
\begin{displaymath}
 \HmtrdAv{X\left|Z_1,Z_2\right.}{[0,1],s'}{\epsilon'} \geqslant
\HmtrdAv{X\left|Z_1\right.}{[0,1],s}{\epsilon} - |Z_2|
\end{displaymath}
\end{theorem}
\noindent In the other approach \cite{Chung2011}, the authors use some \emph{samplability} assumptions. 
\begin{theorem}[\cite{Chung2011}]\label{ChainRule_ifSamplable}
Let $X,Z_1,Z_2$ be as above. Suppose that there exists a random variable $Y'$ with the following properties: (a) $\Hm{Y'|Z_1}\geqslant k$, the pair $(Y',Z_1),(X,Z_1)$ is $(s,\epsilon)$ indistinguishable and (b) there exists a randomized circuit $\Gamma$ of complexity $s_{\Gamma}$, which receives on its input $z\in\mathrm{supp}(Z_1)$ and return samples of $Y'|Z_1=z_1$. Then for $s'= \Omega\left(s\cdot 2^{-m_2}\delta-s_{\Gamma}\right),\, \epsilon' \approx \epsilon+\delta$ we have
\begin{equation*}
\mathbf{H}^{\textup{Metric},[0,1],s',\epsilon'}\left(X|Z_1,Z_2\right) \geqslant \mathbf{H}^{\textup{Metric},[0,1],s,\epsilon}(\left.X\right|Z_1)-\left|Z_2\right| - \log(1/\delta).
\end{equation*}
\end{theorem}
\noindent Finally, there is yet another result related to the chain rule problem, due to \cite{GentryWichs2010}. The authors prove a version of \ref{LeakageLemma_Computational} for using a \emph{nonstandard definition} of Metric Conditional Min-Entropy, which they call the \emph{relaxed computational entropy}. The difference is in Layer (\ref{AbstractEntropy:a}) of the definition: they require $(X,Z)$, to be indistinguishable from all distribution $\left(Y,Z'\right)$ satisfying $\Hm{Y|Z'}\geqslant k$, where-- in comparison to Definition \ref{HILLEntropy_Definition}-- $Z'$ is \emph{not necessarily} equal to $Z$. As observed in \cite{Reyzin2011}, one can easily generalize their approach to prove an ``efficient'' computational version of \ref{ChainRule_GeneralForm} for this definition, with a loss of a factor at most $\mathrm{poly}\left(2^{m_2},\epsilon^{-1}\right)$ in security. It seems however, that in the context of leakage Definition \ref{MetricEntropy_Definition} and \ref{HILLEntropy_Definition} are more suitable, because $Z$ can be what an adversary might have learned about $X$ \cite{Chung2011}; see also the conclusions in \cite{Pietrzak2013}. Being interested in applications in leakage cryptography, we follow the standard definition of the computational entropy in this paper.

\section{Modulus Entropy}

\noindent Our definition is a bit different than Definition \ref{DecomposableEntropy_Fuller} .  
\begin{definition}[Modulus Metric Entropy]\label{ModulusEntropy}
Let $X\in\{0,1\}^{n}$ and $Z\in\{0,1\}^{m}$ be random variables. Given $\epsilon > 0$ and a class of \emph{deterministic boolean} functions $\cD$, we say that \emph{$X$ conditioned on $Z$ has modulus entropy at least $k$ against $\left(\cD,\epsilon\right)$}, and denote it by $\HabsAv{X|Z}{\cD}{\epsilon}\geqslant k$, if for any $D\in\cD$ there exists a random variable $Y\in\{0,1\}^{n}$, satisfying $\widetilde{\mathbf{H}}_{\infty}(Y|Z)\geqslant k$, such that 
\begin{equation}\label{modulus_inequality}
 \mathbf{E}_{z\leftarrow Z}\left| \mathbf{E}_{x\leftarrow (X|Z=z)} D(x,z) - \mathbf{E}_{x\leftarrow (Y|Z=z)} D(x,z)\right| \leqslant \epsilon
\end{equation}
\end{definition}
\noindent The definition above, formulated for the average-case conditional entropy, can be stated also for  the worst-case version, by replacing $\widetilde{\mathbf{H}}_{\infty}$ with $\mathbf{H}_{\infty}$. Using Lemma \ref{Average_WorstCase_Conversion} we obtain immediately a conversion (with some loss) between them:
\begin{lemma}\label{Computational_Average_WorstCase_Conversion}
Suppose that $\HabsAv{X|Z}{\cD}{\epsilon}\geqslant k$. Then $\Habs{X|Z}{\cD}{\epsilon+\delta}\geqslant k-\log(1/\delta)$.
\end{lemma}
\paragraph{Intuitions and motivations behind modulus entropy.}
The only difference between Def. \ref{MetricEntropy_Definition} and Def. \ref{ModulusEntropy} is that they differ in order of the expectation and absolute value signs. Thus, by the triangle inequality, the Modulus Entropy is smaller than Metric Entropy. However, they are not necessarily equal in general. Indeed, for $D$ distinguishing between $(X,Z)$ and $(Y,Z)$ with the advantage no greater than $\epsilon$, contributions to this advantage from particular values of $z$, given by the expressions $\epsilon_D(z)=\mathbf{E}_{x\leftarrow X|Z=z}D(x,z) - \mathbf{E}_{x\leftarrow Y|Z=z}D(x,z)$ can differ in signs. For a few values $z$ we can ''flip" the output of $D(\cdot,z)$ as to ensure that all their contributions are of the same sign; this is however not possible if $Z$ is to long, unless we lose much in efficiency. Thus $\left|\mathbf{E}_{z\leftarrow Z}\epsilon_D(z)\right| \leqslant \epsilon$ does not imply $\mathbf{E}_{z\leftarrow Z}\left|\epsilon_D(z)\right|\leqslant \epsilon$, which is required by inequality (\ref{modulus_inequality}). In comparison to Definition \ref{DecomposableEntropy_Fuller}, our approach is far more general as allow $\epsilon(z)$ and $k(z)$ to be dependent on a chosen $D$. 

We stress that condition \ref{modulus_inequality} is not unnatural. Its purpose is to give much more control over the particular contributions to the advantage, corresponding to the outcomes of $Z$. For instance, Dodis et al. in \cite{Dodis2013} control the average square of the contributions to the advantage (by the inequality $\mathbf{E}_{z\leftarrow Z}\left|\epsilon_D(z)\right|^2\leqslant \epsilon$), defining ''squared indistinguishability".

\subsection{Leakage Chain Rule for Modulus Entropy}

We now show how modulus entropy allows us to prove a leakage chain rule. 
We start with the reformulation of the leakage lemma proved in \cite{Fuller2011}.
\begin{lemma}[Corollary from \cite{Fuller2011}, Proof of Lemma 3.5]\label{LeakageLemma}
Let $D$ be a boolean function, $(X,Z)$ as in Lemma\ \ref{LeakageLemma_Computational}. Suppose $\left| \Ev{D}{X}{x} - \Ev{D}{Y}{x}\right| \leqslant \epsilon$, where $\mathbf{H}_{\infty}(Y)\geqslant k$. Then for any $z\in\mathrm{supp}(Z)$ there exist a distribution $Y'_z$ such that $\Hm{Y'_z} \geqslant k-\log(1/\mathbf{P}_Z(z))$ and $\left| \Ev{D}{X|Z=z}{x} - \Ev{D}{Y'_z}{x}\right| \leqslant \epsilon / \mathbf{P}(Z=z)$.
\end{lemma}
\noindent Now we are in position to prove the following (tight) chain rule. 
\begin{theorem}\label{ChainRule_IfModulus}
Let $X,Z_1,Z_2$ be as in Lemma\  \ref{ChainRule_MinEntropy} and $\cD$ be a class of boolean functions. 
Suppose that $\HabsAv{\left.X\right|Z_1}{\cD}{\epsilon} \geqslant k$. Then $\HabsAv{\left.X\right|Z_1,Z_2}{\cD}{2^{m_2}\epsilon} \geqslant k-m_2$.
\end{theorem}
\begin{proof}
Fix a distinguisher $D = D\left(x,z_1,z_2\right)$. We will construct a distribution $(Y,Z_1,Z_2)$ such that $\HmAv{Y\left|Z_1,Z_2\right.}\geqslant k-m_2$ and $D$ cannot distinguish $(X,Z_1,Z_2)$ from $(Y,Z_1,Z_2)$ with advantage better than $2^{m_2}\epsilon$. For any $z_2$, let $\left(Y^{z_2},Z_1\right)$ be a distribution corresponding to $D\left(\cdot,z_2\right)$ in the sense of Definition \ref{ModulusEntropy} (we write $Y^{z_2}$ as this distribution depends also on $z_2$).  More precisely, $\left(Y^{z_2},Z_1\right)$ is such that
\begin{equation}\label{eq:precis}
 \mathbf{E}_{z_1\leftarrow Z_1} \underbrace{\left|\mathbf{E}_{x\leftarrow (X|Z_1=z_1)} D(x,z_1,z_2) - \mathbf{E}_{x\leftarrow (Y^{z_2}|Z_1=z_1)} D(x,z_1,z_2)\right|}_{\epsilon_D(z_1,z_2) :=} \leqslant \epsilon
\end{equation}
holds (cf.\ (\ref{modulus_inequality}) in Definition \ref{ModulusEntropy}).  For every pair $(z_1,z_2)$ let $\epsilon_D(z_1,z_2)$ denote the value within the first expected value sign, as indicated on (\ref{eq:precis}).
Now, Lemma \ref{LeakageLemma}  implies that for any $z_1,z_2$ there exists a distribution $Y'_{z_1,z_2}$ such that 
\begin{equation}\label{eq:3.2}
 \left| \mathbf{E} D\left(X\left|Z_1=z_1,Z_2=z_2\right.,z_1,z_2 \right) - \mathbf{E} D\left(Y'_{z_1,z_2},z_1,z_2 \right) \right| \leqslant \frac{\epsilon_D\left(z_1,z_2\right)}{\mathbf{P}_{Z_2|Z_1=z_1}\left(z_2\right)}
\end{equation}
and its min-entropy $\Hm{Y'_{z_1,z_2}}$ is at least  $k\left(z_1,z_2\right)$, where 
\begin{equation}\label{eq:3.3}
 k\left(z_1,z_2\right) 
 \geqslant  \mathbf{H}_{\infty}\left(\left. Y^{z_2}\right| {Z_1=z_1}\right)-\log(1/\mathbf{P}_{Z_2|Z_1=z_1}\left(z_2\right))
\end{equation}
Define $\left(Y,Z_1,Z_2\right)$ by $(Y\left|Z_1=z_1,Z_2=z_2\right.) \eqd Y'_{z_1,z_2}$. Now we have \begin{eqnarray*}
   &&\mathbf{E}_{\left( z_1,z_2 \right) \leftarrow \left(Z_1,Z_2\right)} \overbrace{\left| \mathbf{E}_{x\leftarrow X\left|Z_1=z_1,Z_2=z_2\right.}D\left(x,z_1,z_2\right) -\mathbf{E}_{x\leftarrow Y'_{z_1,z_2}}D\left(x,z_1,z_2\right) \right|}^{\leqslant \frac{\epsilon_D\left(z_1,z_2\right)}{\mathbf{P}\left(Z_2=z_2|Z_1 = z_1\right)} \mathrm{\ (by (\ref{eq:3.2}))}}  \\
   & \leqslant & \sum\limits_{z_1,z_2}\mathbf{P}((Z_1,Z_2) = (z_1,z_2)) \cdot \frac{\epsilon_D\left(z_1,z_2\right)}{\mathbf{P}\left(Z_2=z_2|Z_1 = z_1\right)} \\
& =  & \sum\limits_{z_1, z_2}\mathbf{P}\left(Z_1=z_1\right)\epsilon_D\left(z_1,z_2 \right) =  \sum\limits_{z_2}\mathbf{E}_{z_1\leftarrow Z_1}\epsilon_D\left(z_1,z_2\right) \leqslant \sum\limits_{z_2} \epsilon = 2^{m_2}\epsilon
\end{eqnarray*}
where the last inequality follows from (\ref{eq:precis}).  It remains to prove that \\ $\HmAv{Y\left|Z_1,Z_2\right.}\geqslant k-m_2$. We have:
\begin{align*}
 \mathbf{E}_{\left(z_1,z_2\right)\leftarrow \left(Z_1,Z_2\right)} 2^{-k\left(z_1,z_2\right)} & \leqslant   \mathbf{E}_{\left(z_1,z_2\right)\leftarrow \left(Z_1,Z_2\right)} \left[ \frac{ \max\limits_{x}\mathbf{P}\left[ \left.Y^{z_2}=x\right|Z_1=z_1\right] }{\mathbf{P}_{Z_2|Z_1=z_1}( z_2 )} \right] \\
&=  \sum\limits_{z_1,z_2} \max\limits_{x}\mathbf{P}\left[ \left.Y^{z_2}=x\right|Z_1=z_1\right]\cdot \mathbf{P}_{Z_1} (z_1) & \\ 
&=  \sum\limits_{z_2} \mathbf{E}_{z_1\leftarrow Z_1}\left[ \max\limits_{x}\mathbf{P}\left[\left. Y^{z_2}=x \right| Z_1=z_1 \right] \right] \leqslant 2^{m_2}\cdot 2^{-k}
\end{align*}
where the first step follows from (\ref{eq:3.3}) and the last one from $\widetilde{\mathbf{H}}_{\infty}\left(\left.Y^{z_2}\right|Z_1\right)\geqslant k$. 
\end{proof}
\begin{remark}
Note that the entropy obtained after the leakage is again  the Modulus Entropy. Thus, one can apply this theorem several times. The samplability restriction in Thm. \ref{ChainRule_ifSamplable} does not have this feature.
\end{remark}

\paragraph{Chain Rule for entropy against different circuits classes}
Theorem \ref{ChainRule_IfModulus} deals only with entropy against boolean deterministic distinguishers $\cD$. It is natural to ask if one could replace this class with a more general one, in particular, would the theorem still hold if, in its statement, $\cD$ was equal to the class of randomized or real-valued distinguishers.   We answer this question affirmatively in   Lemma \ref{lemma:general} below.  To make its statement as strong as possible, in the assumption we use the Modulus Entropy against boolean deterministic circuits as the weakest option and the HILL Entropy as the strongest notion in the assertion.\footnote{Recall that for the HILL Entropy all kinds of circuits: deterministic boolean, deterministic real valued, randomized boolean are equivalent \cite{Fuller2011} thus we can abbreviate the notation writing just $\mathbf{H}^{\textup{HILL},s',\epsilon'}\left(X|Z \right)$. }
\begin{lemma}\label{lemma:general}
Let $X,Z$ be as in Theorem \ref{LeakageLemma_Computational}.
Suppose that $\HabsAv{X|Z}{s}{\epsilon} \geqslant k$. Then $\mathbf{H}^{\textup{HILL},s',\epsilon'}\left(X|Z \right)\geqslant k'$, where 
$\epsilon' = \epsilon + 2\delta$, $s' = s\cdot \mathcal{O}\left( \frac{\delta^2}{n+m}\right), k'=k-\log\frac{1}{\delta}$.
\end{lemma}
\begin{proof}
If $\HabsAv{X|Z}{s}{\epsilon} \geqslant k$ then $\HmtrAv{X|Z}{\textup{det}\{0,1\},s}{\epsilon}\geqslant k$, as
we pointed out in the discussion after Lemma \ref{Computational_Average_WorstCase_Conversion}. Lemma \ref{Computational_Average_WorstCase_Conversion} yields $\Hmtr{X|Z}{\textup{det}\{0,1\},s}{\epsilon+\delta}\geqslant k-\log\frac{1}{\delta}$. Since for the Metric Worst-Case Entropy it makes no significant difference whether we use real-valued or boolean distinguishers (see Theorem \ref{nonaverage_passing_to_realvalued} in Appendix \ref{app:TechnicalResults}), we obtain $\widetilde{\mathbf{H}}^{\textup{Metric},s',[0,1],\epsilon+\delta} (X|Z)\geqslant k-\log\frac{1}{\delta}$ where $s'= s+\mathcal{O}(1)$. The claim follows now from Theorem \ref{MetricHILL_Conversion}.
\end{proof}


\section{Passing to Modulus Entropy}

While the modulus entropy, as shown in Theorem \ref{ChainRule_IfModulus}, solves the leakage chain rule problem, it keeps being rather a technical assumption. We will give some concrete examples where its definition is fulfilled, and thus admits the chain rule. In comparison to the assertion of Theorem \ref{ChainRule_IfModulus}, they rely on some other assumptions added to the Metric Entropy of $X|Z$. Conversion to the modulus entropy, with estimated loss in parameters, is summarized in Table \ref{tbl:Conversions}. 
\begin{table}[ht]\label{tbl:Conversions}
\centering
{\small
\scalebox{0.9}
{
\begin{tabular}{| l | l | l | l | l | l |} \hline 
{\tabincell{l}{\\ Additional assumptions on \\ $\widetilde{\mathbf{H}}^{\textup{Metric},\{0,1\},s,\epsilon}\left(X|Z\right) \geqslant k$}}  & \multicolumn{4}{|c|}{Our conversion: $\HabsAv{X|Z}{s'}{\epsilon'}\geqslant k'$} \\ 
 & $k'$ & $\epsilon'$ & $s'$ &\\ \hline 
\tabincell{l}{(a) Decomposabe entropy \cite{Fuller2011}}  & $k$ & $\epsilon$ & $s$ & Thm.\ \ref{Decomposable_implies_Modulus} \\ \hline 
\tabincell{l}{(b) Samplability of $Y|Z=z$ given $z$, \\ where $(Y,Z) \sim^{\epsilon,s} (X,Z)$  \cite{Chung2011}}  & $k-\mathcal{O}\left(\log\frac{1}{\epsilon}\right)$ & $\mathcal{O}\left({\epsilon}\right)$ &  $\mathcal{O}\left(s /\frac{1}{\epsilon^2}\right)$ & Thm.\ \ref{Samplable_implies_Modulus}  \\ \hline 
 \tabincell{l}{(c) Entropy against $\mathrm{poly(n)}$-circuits,  given \\ an access to an $\NP$ oracle over $\{0,1\}^n$}  & $k-\mathcal{O}\left(\log\frac{1}{\epsilon}\right)$ & $\mathcal{O}\left(\sqrt{\epsilon}\right)$ & $\mathcal{O}\left(s/\mathrm{poly}\left(n,\,\frac{1}{\epsilon}\right) \right)$ & \tabincell{l}{Thm.\ \ref{Counting_implies_Modulus}\\ (point \ref{Counting_implies_Modulus2})} \\ \hline
 \tabincell{l}{(d) Entropy very high, \\ i.e. $k > n-\mathcal{O}\left(\log\frac{1}{\epsilon}\right)$} & $k-\mathcal{O}\left(\log\frac{1}{\epsilon}\right)$ & $\mathcal{O}\left({\sqrt{\epsilon}}\right)$ & $\mathcal{O}\left(s / \frac{m+n}{\epsilon^3}\log\frac{1}{\epsilon} \right)$  & \tabincell{l}{Thm.\ \ref{Counting_implies_Modulus}\\ (point \ref{Counting_implies_Modulus1})}  \\ \hline 
 (e) None & $k$ & $2^{t}\epsilon $ & $s-\mathcal{O}\left(2^{m-t}m\right)$ & Thm.\ \ref{Metric_implies_Modulus} \\ \hline   
  \tabincell{l}{ (f) $X$ is $(\epsilon,s)$ squared-indistinguishable \\ from $Y$ given $Z$, and $\HmAv{Y|Z}\geqslant k$ } & $k$ & $\sqrt{\epsilon} $ & $s$ & Thm.\ \ref{SquaredInd_implies_Modulus} \\ \hline   
\end{tabular}} \label{table:1}
}\caption{Conversions to the modulus entropy}
\end{table}
\newline
\noindent As shown, some of these assumptions were already studied in the leakage literature. The proofs of conversions will be given in the next section.

\subsection{Benefits of using Modulus Entropy}

To summarize, let us mention the three key features of the modulus entropy:
\begin{enumerate}[(a)]
\item it implies the metric entropy which is widely used in the leakage-resilient cryptography,
\item it allows to apply the tight chain rule multiple times, and
\item it can be obtained from the known assumptions that guarantee the chain rule (decomposability, samplability) and from other important or at least non-trivial cases  (squared-indistinguishability, high pseudo-entropy, NP-oracle)
\end{enumerate}

\paragraph{Modulus Entropy vs Samplability and Decomposability}
Comparing the conversion results in the table with Theorems \ref{ChainRule_ifDecomposable} and \ref{ChainRule_ifSamplable}, we see that Modulus Entropy is a weaker restriction and still guarantees the chain rule with at least comparable quality. Starting from decomposability or samplability, converting to the Modulus Entropy first and applying the chain rule next (and possibly translating into the HILL entropy) yields the same loss as the direct use of that assumptions.

\section{Proofs of Conversion Results}

Throughout all the proofs in this section, $X,Z$ are random variables over $\{0,1\}^{n}$ and $\{0,1\}^m$ respectively. The proofs are based on the following technical lemma.

\begin{lemma}\label{ModulusEntropy_main_lemma}
Let $X,Z$ be random variables over $\{0,1\}^{n},\{0,1\}^m$.  Suppose that $D$ is such that for all distributions $(Y,Z)$ with $\Hm{Y|Z}\geqslant k$ the following holds:
\begin{equation}\label{eq:above}
 \mathbf{E}_{z\leftarrow Z}\left| \EvAdd{D}{X|Z=z}{x}{z} - \EvAdd{D}{Y|Z=z}{x}{z}\right| \geqslant \epsilon.
\end{equation}
 Then either for $D' = D$ or for $D' = D^c$ we have that for all distributions $(Y,Z)$ with $\mathbf{H}_{\infty}\left(Y|Z\right)\geqslant k$ the following is true:
\begin{equation*}\label{eq:keyestimate}
  \mathbf{P}_{(x,z)\leftarrow (X,Z)}  \left[ D'(x,z) - \EvAdd{D'}{Y|Z=z}{x}{z} \geqslant \epsilon/4 \right] \geqslant {\epsilon^2}/{16}.
\end{equation*} 
\end{lemma}

\begin{proof}
Consider the distribution $(Y^{+},Z)$ which minimizes the left-hand side of (\ref{eq:above}). Define $
 \epsilon(z) := \left| \EvAdd{D}{X|Z=z}{x}{z} - \EvAdd{D}{\left.Y^{+}\right|Z=z}{x}{z} \right|$.
Observe that
\begin{multline*}\label{eq:techlemma:1}
 \min_{(Y,Z):\,\mathbf{H}_{\infty}(Y|Z)\geqslant k}\mathbf{E}_{z\leftarrow Z}\left| \EvAdd{D}{X|Z=z}{x}{z} - \EvAdd{D}{Y|Z=z}{x}{z} \right| = \\ =
\mathbf{E}_{z\leftarrow Z}\left[ \min_{Y_z:\,\mathbf{H}_{\infty}(Y_z)\geqslant k}\left|\EvAdd{D}{X|Z=z}{x}{z} - \EvAdd{D}{Y_z}{x}{z} \right| \right].
\end{multline*}
Therefore, for every distribution $Y_z$ with min-entropy $\Hm{Y_z}\geqslant k$ we have
\begin{equation*}\label{eq:techlemma:2} 
 \left| \EvAdd{D}{X|Z=z}{x}{z} - \EvAdd{D}{Y|Z=z}{x}{z} \right| \geqslant \epsilon(z)
\end{equation*}
Note that if $\epsilon(z)>0$ then either (a) $\EvAdd{D}{X|Z=z}{x}{z} - \EvAdd{D}{Y|Z=z}{x}{z} \geqslant \epsilon(z)$ or (b)
$\EvAdd{D}{X|Z=z}{x}{z} - \EvAdd{D}{Y|Z=z}{x}{z}\leqslant -\epsilon(z)$ holds for all $Y_z$ with $\mathbf{H}_{\infty}\left(Y_z\right)\geqslant k$. This follows from the convexity of the set of distributions $\mathbf{H}_{\infty}\left(Y_z\right)\geqslant k$, which in turn implies that all values of $\mathbf{E}_{x\leftarrow X|z=z}D\left(x,z\right)-\mathbf{E}_{x\leftarrow Y_z}D\left(x,z\right)$, over the choice of $Y_z$, form a convex set. Therefore 
\begin{equation*}\label{eq:techlemma:3} 
\mathbf{E}_{x\leftarrow X|z=z}D'\left(x,z\right)-\mathbf{E}_{x\leftarrow Y_z}D'\left(x,z\right)\geqslant \epsilon(z)  
\end{equation*}
holds for all $Y_z$ with $\mathbf{H}_{\infty}\left(Y_z\right)\geqslant k$, where $D'$ is defined, depending on $z$, by
\begin{equation}\label{eq:D}
D'(x,z) :=
\left\{
\begin{array}{ll}
   D(x,z) & \mbox{in case (a)}\\
   D^c(x,z) & \mbox{in case (b)}\\
   0 & \mbox{if $\epsilon(z) = 0$}.
\end{array}
\right.
\end{equation}
Since $\epsilon(z)\geqslant \frac{\epsilon}{2}$ holds\footnote{Throughout the proofs, we will make use of the simple Markov-style principle: let $X$ be a non-negative random variable bounded by $M$. Then $X > \frac{1}{2M}\mathbf{E}X$ with probability at least $\frac{1}{2}\mathbf{E}X$.} with probability at least $\frac{\epsilon}{2}$ over $z\leftarrow Z$, we get 
\begin{equation*}
  \EvAdd{D'}{X|Z=z}{x}{z}- \max\limits_{Y_z:\, \Hm{Y_z}\geqslant k} \EvAdd{D'}{Y|Z=z}{x}{z} \geqslant {\epsilon}/{2}
\end{equation*}
with probability at least $\frac{\epsilon}{2}$ over $z\leftarrow Z$. For every such $z$ we obtain 
\begin{equation*}
 \mathbf{P}_{x\leftarrow X|Z=z}\left[ D'(x,z) - \max\limits_{Y_z:\, \Hm{Y_z}\geqslant k}\EvAdd{D'}{Y|Z=z}{x}{z}  \geqslant \frac{\epsilon}{4} \right] \geqslant \frac{\epsilon}{4}
\end{equation*}
Taking expectation over $z\leftarrow Z$ we conclude that 
\begin{equation*}
 \mathbf{P}_{(x,z)\leftarrow (X,Z)}\left[ D'(x,z) - \max\limits_{Y_z:\, \Hm{Y_z}\geqslant k}\mathbf{E}D'\left(Y_z,z\right)  \geqslant \frac{\epsilon}{4} \right] \geqslant \frac{\epsilon^2}{8}.
\end{equation*}
Therefore, either for $D' = D$ or $D' = D^c$ the probability on the left-hand side of the above inequality needs to be at least $ \frac{1}{2} \cdot \frac{\epsilon^2}{8}= \frac{\epsilon^2}{16}$, which proves the claim.
\end{proof}
\subsection{Decomposable entropy}
We start by noticing that Definition \ref{DecomposableEntropy_Fuller} is stronger than our Definition \ref{ModulusEntropy}.

\begin{theorem}\label{Decomposable_implies_Modulus}
Suppose $\HmtrdAv{X|Z}{s}{\epsilon}\geqslant k$. Then $\HabsAv{X|Z}{s}{\epsilon}\geqslant k$. 
\end{theorem}

\begin{proof}
Fix a distinguisher $D = D(x,z)$. According to Def. \ref{DecomposableEntropy_Fuller}, for every $z$ we have a distribution $Y_z$ such that $\Hm{Y_z}\geqslant k(z)$ and $\left|\mathbf{E}D(X|Z=z)-\mathbf{E}D(Y_z)\right|\leqslant \epsilon(z)$. Consider a distribution $(Y,Z)$ defined by $\left(Y|Z=z\right)\eqd Y_z$. Since $\mathbf{E}_{z\leftarrow Z}\epsilon(z) \leqslant \epsilon$, we obtain inequality (\ref{modulus_inequality}). In turn, the assumptions on $k(z)$ implies $\HmAv{Y|Z}\geqslant k$.
\end{proof}

\noindent The following theorem converts Metric Entropy into Modulus Entropy (cf. case (e) in Table \ref{table:1}). Its principal significance is that the equivalence between both definitions is established, provided that $Z$ is sufficiently short (grows at most logarithmically in the security parameters). 

\begin{theorem}\label{Metric_implies_Modulus}
Suppose that $\mathbf{H}^{\textup{Metric},\{0,1\},s,\epsilon}(X|Z)\geqslant k$. Then $\Habs{X|Z}{s'}{\epsilon'}\geqslant k$, where $\epsilon' = 2^{t}\epsilon$ and $s'=s-\mathcal{O}\left(2^{m-t}m\right)$.
\end{theorem}

\begin{proof}
For the sake of contradiction suppose that for some $D$ of complexity $s'$ and for every $(Y,Z)$ such that $\mathbf{H}_{\infty}(Y|Z)\geqslant k
$ we have that 
\begin{displaymath}
\mathbf{E}_{z\leftarrow Z}\left|\mathbf{E}_{x\leftarrow X|Z=z}D\left(x,z\right)-\mathbf{E}_{x\leftarrow Y|Z=z} D\left(x,z\right) \right|\geqslant \epsilon'.
\end{displaymath}
Applying the same reasoning as at the beginning of the proof of Lemma \ref{ModulusEntropy_main_lemma}, we obtain that there exist a distinguisher $D'$ (cf.\ (\ref{eq:D})) such that for every distribution $Y_z$ with $\Hm{Y_z}\geqslant k$ it holds that
\begin{equation}\label{eq:trivial}
\mathbf{E}_{x\leftarrow X|z=z}D'\left(x,z\right)-\mathbf{E}_{x\leftarrow Y_z}D'\left(x,z\right)\geqslant \epsilon'(z),  
\end{equation}
where $\mathbf{E}_{z\leftarrow Z}\epsilon'(z)\geqslant \epsilon'$. Thus, for every  $(Y,Z)$ such that $\mathbf{H}_{\infty}\left(Y|Z\right)\geqslant k$ we have
\begin{equation*}
 \mathbf{E}_{(x,z)\leftarrow (X,Z)}D'\left(x,z\right) -\mathbf{E}_{(x,z)\leftarrow (Y,Z)}D'\left(x,z\right) \geqslant \mathbf{E}_{z\leftarrow Z}\epsilon'(z) \geqslant \epsilon'.
\end{equation*}
Recall that in the proof of Lemma \ref{ModulusEntropy_main_lemma}, the value $D'(x,z)$ is defined as equal to $D(x,z)$ or $D^{c}(x,z)$ or $0$, depending on $z$. Instead, we can follow that construction with respect to only $2^{m-t}$ ``heaviest'' values $z$ maximizing $\mathbf{P}(Z=z)\epsilon'(z)$ and setting $D'=0$ for other $z$. The obtained circuit is of size at most $s'+\mathcal{O}\left(2^{m-t}m\right) = s$ and distinguishes with the advantage at least $2^{-t}\epsilon' = \epsilon$.
\end{proof}

\subsection{The samplability assumption}

\noindent In the next theorem we deal with the samplability assumption used in \cite{Chung2011}.

\begin{theorem}\label{Samplable_implies_Modulus}
Suppose that $(X,Z)$ is $(s,\epsilon)$-indistinguishable from a distribution $(Y',Z)$, with the following properties (a) $\Hm{Y'|Z}\geqslant k$ and (b) there exists a randomized circuit $\Gamma$ receiving on its input $z\in\mathrm{supp}(Z)$ and returning samples from the distribution $Y'|Z=z$. Then
\begin{equation*}
 \Habs{X|Z}{s\cdot \frac{\epsilon^2}{64}-\mathrm{size}(\Gamma)}{8\sqrt{\epsilon}} \geqslant k - 2\log\left(1/\epsilon\right)-7.
\end{equation*}
\end{theorem}
\begin{proof}
Suppose that $\Habs{X|Z}{s'}{\epsilon'} < k'$, where $k'=k - 2\log\left({1}/{\epsilon}\right)-7$ and $\epsilon' = {\epsilon^2}/{64}$ and $s'={\epsilon^2s}/{64}-\mathrm{size}(\Gamma)$. Thus, for some $D$ of size $s'$ and every $(Y,Z)$ with $\mathbf{H}_{\infty}(Y|Z)\geqslant k'$ we have
\begin{equation}\label{eq:modulus_entropy_negation}
 \mathbf{E}_{z\leftarrow Z}\left| \EvAdd{D}{X|Z=z}{x}{z} - \EvAdd{D}{Y|Z=z}{x}{z}\right| \geqslant \epsilon'.
\end{equation}
Let $D'$ be a distinguisher obtained from Lemma \ref{ModulusEntropy_main_lemma}. Consider the following $D''$: on input $(x,z)$, which comes either from $(X,Z)$ or $(Y',Z)$ do the following:
\begin{itemize}
  \item for $i = 1$ to $\ell = \left\lceil \frac{64}{\epsilon^2}\right\rceil-1$ sample $y_i \leftarrow Y'|Z=z$  using the circuit $\Gamma$,
    \item if $D'(x,z) > \max\limits_{i=1,\ldots,l} D'\left(y_i,z\right)$ --- output $1$, otherwise output $0$.
\end{itemize}
Clearly $D''$ has complexity at most $\left(\ell+1\right) \cdot \left(s'+\mathrm{size}(\Gamma)\right)=s$. We will show that it gives sufficient distingishing advantage. We start with the following easy observation, used implicitly in \cite{Chung2011} (the proof of Lemma 16).  

\begin{lemma}\label{extreme_distribution_inequality}
For $D$ be a $[0,1]$-valued function. If $Y^{+}$ is distributed uniformly over $\Max^k_{D}$, then for any $Y$ with $\mathbf{H}_{\infty}(Y) \geqslant k + \log\frac{1}{\delta}$ we have
\begin{equation*}
 \mathbf{P}_{x \leftarrow Y}\left[  D(x) - \mathbf{E}_{x\leftarrow Y^{+}}D(x) > 0  \right] < \delta.
\end{equation*}
\end{lemma}
\noindent The proof that $D''$ is indeed a good distinguisher consists of two steps 

\begin{claim}
On input $(x,z)\leftarrow (X,Z)$ the circuit $D''$ outputs $1$ w.p. at least $\epsilon'^2/32$. 
\end{claim}
\begin{proof}
Consider a distribution $\left(Y^{+},Z\right)$ such that for every $z$ the distribution $\left. Y^{+}\right|Z=z$ is uniform over $\Max^{k}_{D(\cdot,z)}$. Since $y_i$ are independent and distributed according to $Y'|Z=z$, it follows from Lemma \ref{extreme_distribution_inequality} that $\Ev{D'}{\left.Y^{+}\right|Z=z}{x} \geqslant \max\limits_{i} D'(y_i,z)$ holds with probability at least $\left(1-{2^{k'-k}}\right)^{\ell} \geqslant 1 - \ell \cdot 2^{k'-k} \geqslant \frac{1}{2}$. Now, Lemma \ref{ModulusEntropy_main_lemma} yields  $D'(x,z) > \Ev{D'}{\left.Y^{+}\right|Z=z}{x}$ with probability at least $\frac{\epsilon'^2}{16}$ over $(x,z)$. Since sampling $y_i$ is independent from $(X,Z)$, the claim follows.
\end{proof}

\begin{claim}
On input $(y,z)\leftarrow (Y',Z)$ the circuit $D''$ outputs $1$ w.p. at most $\epsilon'^2/64$.  
\end{claim}
\begin{proof}
Note that $y$ and $y_1,\ldots,y_\ell$  are all independent copies of the distributions $Y'|Z=z$. Therefore probability that $y > \max_{i=1,\ldots,l} y_i$ is at most $\frac{1}{\ell+1} \leqslant \frac{\epsilon'^2}{64}$.
\end{proof}
\noindent From the last two claims we get 
$ \mathbf{P}\left( D''\left(X,Z\right) = 1 \right) - \mathbf{P}\left( D''\left(Y,Z\right) = 1 \right) \geqslant \epsilon'^2/64$, which completes the proof of Theorem \ref{Samplable_implies_Modulus}.
\end{proof}

\subsection{Approximate counting}
It turns out that using a technique called the {\em approximate counting}, one can show a conversion from metric to modulus entropy. However, we need some additional assumptions to achieve both: high accuracy and efficiency in counting:
\begin{theorem}\label{Counting_implies_Modulus}
Suppose that one of the following is true:
\begin{enumerate}[{} (a)]
 \item \label{Counting_implies_Modulus1} $\mathbf{H}_{}^{\textup{Metric},\mathrm{rand}\{0,1\},s,\epsilon}(X|Z) \geqslant k$ against circuits of size $s$,
 \item \label{Counting_implies_Modulus2} $\mathbf{H}_{}^{\textup{Metric},\{0,1\},s,\epsilon}(X|Z) \geqslant k$ against circuits of size $s=\mathrm{poly}(n)$, with an access to an $\NP$-oracle.
\end{enumerate}
Then we have $\Habs{X|Z}{s'}{\epsilon'} \geqslant k'$, where $\epsilon' = 8\sqrt{\epsilon}$, $k' =  k-\log\frac{1}{\epsilon}$ and $s'$ given by $s'=\mathcal{O}\left(s \cdot \frac{2^{k-n-2}\cdot \epsilon}{\log(1/\epsilon)} \right)$ in case (\ref{Counting_implies_Modulus1}) or $s'={\mathrm{poly}\left(n,\epsilon\right)}$ in case (\ref{Counting_implies_Modulus2}).
\end{theorem}
\noindent Note that to make the conversion in (a) efficient, we need the assumption that $k$ is large as it is easy to see that if $k$ is much smaller than $n$ then, in the formula that gives the bound on $s'$, the $2^{k-n-2}$ factor starts to dominate over $\epsilon$.
\begin{proof}[Proof of Theorem \ref{Counting_implies_Modulus}]
Suppose that $\Habs{X|Z}{s'}{\epsilon'} < k'$. Then Lemma \ref{ModulusEntropy_main_lemma} implies that for all $Y\in\{0,1\}^n$ with $\mathbf{H}_{\infty}(Y|Z)\geqslant k'$ and some distinguisher $D'$ of complexity $s'+1$ we have
\begin{equation}\label{niewiem}
  \mathbf{P}_{(x,z)\leftarrow (X,Z)} \left[ D'(x,z) - \mathbf{E}D'\left(Y|Z=z,z\right) \geqslant \epsilon'/4 \right] \geqslant \epsilon'^2/16.
\end{equation} 
Since 
\begin{equation}\label{extreme_disttribution_description}
 \max\limits_{Y_z:\,\Hm{Y_z}\geqslant k'}\mathbf{E}D\left(Y_z,z\right) = \min\left(1,\,2^{-k'}\left|D'(\cdot,z)\right| \right)
\end{equation}
hence, combining this with (\ref{niewiem}),  we obtain
\begin{equation}\label{extreme_distribution_inequality2}
  \mathbf{P}_{(x,z)\leftarrow (X,Z)} \left[  D'(x,z) - 2^{-k'}\left|D'\left(\cdot,z\right)\right| \geqslant \epsilon'/4 \right] \geqslant \epsilon'^2/16.
\end{equation} 
We now show that there exists a randomized function $h$ such that for every $z$ 
\begin{equation}\label{approximated_value}
 \mathbf{P}\left( \left| h(z) - 2^{-k'}\left|D'\left(\cdot,z\right)\right| \right| \leqslant \epsilon'/8 \right) \geqslant 1-\epsilon'^2/64,
\end{equation}
and $h(z)$ is samplable for all $z$'s satisfying $\left|D'\left(\cdot,z\right)\right| < 2^{k'}$.  More precisely: there  exists a randomized circuit of size $\mathcal{O}\left( s' \cdot \frac{2^{n-k}}{\epsilon^2}\log\frac{1}{\epsilon }\right)=s$, which computes $h(z)$ correctly for every such $z$.  This is a corollary from the following claim.
\begin{claim}\label{Chernoff}
Let $D$ be a boolean circuit such that $|D|\leqslant 2^{k}$. Then for $\delta',\delta'' \in \left(0, \frac{1}{2}\right)$, $\ell > 4\cdot 2^{n-k}\frac{1}{\delta'^2}\log\frac{1}{\delta''}$ and $\left(U_i\right)_{i=1,\ldots,\ell}$ being independent and uniform over $\{0,1\}^{n}$,  the following inequality holds:
\begin{equation*}
\mathbf{P}\left[ \left| \ell^{-1}\sum_{i=1}^{\ell}D\left(U_i\right) - 2^{-n}|D| \right| \geqslant 2^{k-n}\delta' \right] \leqslant  2 \delta''. 
\end{equation*}
\end{claim}
\begin{proof}
Define $g=\frac{1}{\ell}\sum_{i=1}^{\ell}D\left(U_i\right)$. The Chernoff Inequality\footnote{We use the following version: Let $X_i$ be random variables satisfying $\left|X_i-\mathbf{E}X_i\right| \leqslant 1$ and $X=\sum_{i}X_i$. Then $\mathbf{P}\left[ \left|X-\mathbf{E}X\right| \geqslant \lambda\sigma \right]\leqslant 2\min\left( e^{-\frac{\lambda^2}{4}},\,e^{-\frac{\lambda\sigma}{2}} \right)$, where $\sigma = \mathrm{Var}(X)$} yields
\begin{equation*}
 \mathbf{P}\left[ \left| g-\mathbf{E}D(U) \right| \geqslant \delta \right] \leqslant 2\max\left(e^{-\frac{\delta^2 \ell^2}{4\sigma^2}},\,e^{-\frac{\ell\delta}{2}} \right),
\end{equation*}
where $\sigma^2 = \mathrm{Var}\left(\sum_{i=1}^{\ell}D\left(U_i\right)\right)$. Since $\mathrm{Var}\left( D\left(U_i\right) \right)=2^{k-n}\left(1-2^{k-n}\right)$ we have $\sigma^2 =  \ell\cdot 2^{k-n}\left(1-2^{k-n}\right)$. By setting $ 2^{n-k}\delta = \delta'$ we get $\frac{\delta^2 \ell^2}{4\sigma^2} \geqslant \frac{2^{k-n}\ell\delta'^2}{4}$ and $\frac{\ell\delta}{2} \geqslant \frac{2^{k-n}\ell\delta'}{2}$. Since $\mathbf{E}D(U) = |D|/2^{n}$, choosing $\ell$ sufficiently large so that $2^{k-n}\ell\delta'^2 > 4\log(1/\delta'')$, we obtain $
 \mathbf{P}\left[ \left| g\cdot 2^{n-k} - |D|\cdot 2^{-k} \right| \geqslant \delta' \right] \leqslant 2 e^{-\log\frac{1}{\delta''}} < 2 \delta''$.
\end{proof}
\noindent It follows from the claim that $h(z)=\frac{2^{k-n}}{\ell}\sum_{i=1}^{\ell}D\left(U_i,z\right)$ is a required sampler.
Consider the following distinguisher $D''$: on input $(x,z)$, which comes either from $(X,Z)$ or $(Y,Z)$, return $1$ iff 
$D'(x,z) > h(z) + \frac{\epsilon'}{8}$. We will prove that $D''$ distinguishes between $(X,Z)$ and all $(Y,Z)$ satisfying $\Hm{Y|Z}\geqslant k$. Note that if $D''(x,z)=1$ then $h(z)<1-\frac{\epsilon'}{8}$ and hence $\left|D'\left(\cdot,z\right)\right| < 2^{k'}$. Especially, $D''$ is of complexity at most $s$. Now, inequalities (\ref{approximated_value}) and (\ref{extreme_distribution_inequality2}) yield
\begin{multline*}
 \mathbf{P}_{(x,z)\leftarrow (X,Z)}\left[ D'(x,z) > h(z) + \epsilon'/8 \right] \geqslant \\
\mathbf{P}_{(x,z)\leftarrow (X,Z)}\left[ D'(x,z) > 2^{-k'}\left|D'\left(\cdot,z\right)\right| + \epsilon'/4 \right] - \epsilon'^2/64 \geqslant {3\epsilon'^2}/{64},
\end{multline*}
Let $k' = k+\log(\frac{1}{\delta}$ where $\delta = \frac{\epsilon'^2}{64}$. From (\ref{extreme_disttribution_description}), (\ref{extreme_distribution_inequality2}), (\ref{approximated_value}) and Lemma \ref{extreme_distribution_inequality}, we obtain 
\begin{multline*}
  \mathbf{P}_{(x,z)\leftarrow (Y,Z)}\left[ D'(x,z) > h(z) + {\epsilon'}/{8} \right] \leqslant  \\
\mathbf{P}_{(x,z)\leftarrow (Y,Z)}\left[ D'(x,z) > 2^{-k'}{\left|D'\left(\cdot,z\right)\right|} \right] + {\epsilon'^2}/{64} \leqslant {\epsilon'^2}/{32}.
\end{multline*}
Combining the last two estimates yields, if only $\mathbf{H}_{\infty}(Y|Z)\geqslant k'$, the inequality 
\begin{equation*}
  \mathbf{P}\left[D'(X,Z) = 1\right] - \mathbf{P}\left[D'(Y,Z) = 1\right] \geqslant \epsilon'^2/64
\end{equation*}
which completes the proof for case (\ref{Counting_implies_Modulus1}). In case (\ref{Counting_implies_Modulus2}), we proceed in the same way but we compute numbers $h(z)$ using an $\NP$ oracle. The basic result we use can be stated as follows: 
\begin{lemma}{\cite{Donnell2009}}
There is a probabilistic algorithm which, given a boolean circuit $D$ over $\{0,1\}^{n}$ of size $\mathrm{poly}(n)$ and a natural number $M$, decides, with success probability at least $\frac{3}{4}$, whether $\frac{1}{4} M < |D| < 4 M$, in time $\mathrm{poly}\left(n\right)$, using an oracle for $\NP$.
\end{lemma}
\noindent Let us make three important observations: 
\begin{itemize}
 \item The success probability $\frac{3}{4}$ can be amplified to $1-\delta$, by repeating the algorithm $\mathcal{O}\left(\log\frac{1}{\delta}\right)$ times and taking the majority answer.
 \item The factor $4$ can be improved to $1+\gamma$, by running the algorithm on the circuit $D' = D_1\wedge\ldots\wedge D_k$, where $D_i$ for $i=1,\ldots,k$ are copies of $D$ and $k$ is such that $(1+\gamma)^{k}\leqslant 4$.
\end{itemize}
Hence, there is an algorithm which, with probability at least $1-\delta$, computes a value $g$ such that
$ (1-\gamma) M < |D| < (1+\gamma) M$, in time $\mathrm{poly}\left(n,\frac{1}{\gamma},\log\frac{1}{\delta}\right)$, using an oracle for $\NP$. For every $z$, let $M(z)$ be a value obtained by applying this algorithm to the circuit $D'(\cdot,z)$ and  $\gamma = \frac{\epsilon'}{16}$, $\delta = 1-\frac{\epsilon'^2}{64}$. Define $h(z) := 2^{-k}M(z)$. If $\left|D'(\cdot,z)\right| < 2^{k'}$, then $\left|M(z)-\left|D'(\cdot,z)\right|\right| \leqslant 2\cdot 2^{k'}\cdot\frac{\epsilon'}{16}$ holds with probability at least $1-\frac{\epsilon'^2}{64}$, and thus for such values $z$ holds the same estimate as in (\ref{approximated_value}). We proceed further with $h$ as in the previous proof. 
\end{proof}

\subsection{Squared Indistinguishability}

\begin{theorem}\label{SquaredInd_implies_Modulus}
We say that $X$ is $(s,\epsilon)$ squared-indistinguishable from $Y$ given $Z$, if for every circuit $D$ of size $s$, $\mathbf{E}_{z\leftarrow } \left[ \mathbf{E} D(X|Z=z,z) - \mathbf{E}D(Y|Z=z,z) \right]^2 \leqslant \epsilon$ (motivated by \cite{Dodis2013}). Suppose that $X|Z$ is $(s,\epsilon)$ squared-indistinguishable from $Y$ given $Z$, and $\HmAv{Y|Z}\geqslant k$. Then $\Habs{X|Z}{s}{\sqrt{\epsilon}} \geqslant k$.
\end{theorem}
\begin{proof}
From the inequality between the first and the second moment we obtain:
\begin{multline}
  \mathbf{E}_{z\leftarrow Z}\left| \mathbf{E} D(X|Z=z,z) - \mathbf{E}D(Y|Z=z,z) \right| \leqslant \\ \left( \mathbf{E}_{z\leftarrow Z} \left[ \mathbf{E} D(X|Z=z,z) - \mathbf{E}D(Y|Z=z,z) \right]^2 \right)^{\frac{1}{2}} \leqslant \sqrt{\epsilon}.
\end{multline}
\end{proof}

\section{Acknowledgments}
I would like to express special thanks to Stefan Dziembowski and Krzysztof Pietrzak, for their helpful suggestions and discussions.

\bibliographystyle{amsalpha}
\bibliography{./ComputationalEntropy}

\providecommand{\bysame}{\leavevmode\hbox to3em{\hrulefill}\thinspace}
\providecommand{\MR}{\relax\ifhmode\unskip\space\fi MR }
\providecommand{\MRhref}[2]{%
  \href{http://www.ams.org/mathscinet-getitem?mr=#1}{#2}
}
\providecommand{\href}[2]{#2}
\begin{thebibliography}{RTTV08}

\bibitem[BSW03]{Barak2003}
Boaz Barak, Ronen Shaltiel, and Avi Wigderson, \emph{Computational analogues of
  entropy.}, RANDOM-APPROX (Sanjeev Arora, Klaus Jansen, José D.~P. Rolim, and
  Amit Sahai, eds.), Lecture Notes in Computer Science, vol. 2764, Springer,
  2003, pp.~200--215.

\bibitem[CKLR11]{Chung2011}
Kai-Min Chung, Yael~Tauman Kalai, Feng-Hao Liu, and Ran Raz, \emph{Memory
  delegation}, Cryptology ePrint Archive, Report 2011/273, 2011,
  \url{http://eprint.iacr.org/}.

\bibitem[DORS08]{Dodis2008}
Yevgeniy Dodis, Rafail Ostrovsky, Leonid Reyzin, and Adam Smith, \emph{Fuzzy
  extractors: How to generate strong keys from biometrics and other noisy
  data}, SIAM J. Comput. \textbf{38} (2008), no.~1, 97--139.

\bibitem[DP08]{Dziembowski2008}
Stefan Dziembowski and Krzysztof Pietrzak, \emph{Leakage-resilient cryptography
  in the standard model}, IACR Cryptology ePrint Archive \textbf{2008} (2008),
  240.

\bibitem[DY13]{Dodis2013}
Yevgeniy Dodis and Yu~Yu, \emph{Overcoming weak expectations}, Proceedings of
  the 10th theory of cryptography conference on Theory of Cryptography (Berlin,
  Heidelberg), TCC'13, Springer-Verlag, 2013, pp.~1--22.

\bibitem[FOR12]{FullerReyzin2012}
Benjamin Fuller, Adam O'Neill, and Leonid Reyzin, \emph{A unified approach to
  deterministic encryption: New constructions and a connection to computational
  entropy}, Cryptology ePrint Archive, Report 2012/005, 2012,
  \url{http://eprint.iacr.org/}.

\bibitem[FR12]{Fuller2011}
Benjamin Fuller and Leonid Reyzin, \emph{Computational entropy and information
  leakage}, Cryptology ePrint Archive, Report 2012/466, 2012,
  \url{http://eprint.iacr.org/}.

\bibitem[GW10]{GentryWichs2010}
Craig Gentry and Daniel Wichs, \emph{Separating succinct non-interactive
  arguments from all falsifiable assumptions}, Cryptology ePrint Archive,
  Report 2010/610, 2010, \url{http://eprint.iacr.org/}.

\bibitem[HILL99]{HILL99}
Johan Hastad, Russell Impagliazzo, Leonid~A. Levin, and Michael Luby, \emph{A
  pseudorandom generator from any one-way function}, SIAM J. Comput.
  \textbf{28} (1999), no.~4, 1364--1396.

\bibitem[HLR07]{Reyzin2007}
Chun-Yuan Hsiao, Chi-Jen Lu, and Leonid Reyzin, \emph{Conditional computational
  entropy, or toward separating pseudoentropy from compressibility},
  Proceedings of the 26th annual international conference on Advances in
  Cryptology (Berlin, Heidelberg), EUROCRYPT '07, Springer-Verlag, 2007,
  pp.~169--186.

\bibitem[KPW13]{Pietrzak2013}
Stephan Krenn, Krzysztof Pietrzak, and Akshay Wadia, \emph{A counterexample to
  the chain rule for conditional hill entropy}, Theory of Cryptography (Amit
  Sahai, ed.), Lecture Notes in Computer Science, vol. 7785, Springer Berlin
  Heidelberg, 2013, pp.~23--39.

\bibitem[OG09]{Donnell2009}
Ryan O'Donnell and Venkatesan Guruswami, \emph{An intensive introduction to
  computational complexity theory}, University Lecture, 2009,
  \url{http://www.cs.cmu.edu/~odonnell/complexity/}.

\bibitem[Rey11]{Reyzin2011}
Leonid Reyzin, \emph{Some notions of entropy for cryptography}, Information
  Theoretic Security (Serge Fehr, ed.), Lecture Notes in Computer Science, vol.
  6673, Springer Berlin Heidelberg, 2011, pp.~138--142.

\bibitem[RTTV08]{Reingold2008}
Omer Reingold, Luca Trevisan, Madhur Tulsiani, and Salil Vadhan, \emph{Dense
  subsets of pseudorandom sets}, Proceedings of the 2008 49th Annual IEEE
  Symposium on Foundations of Computer Science (Washington, DC, USA), FOCS '08,
  IEEE Computer Society, 2008, pp.~76--85.

\bibitem[Sha48]{Shannon1948}
C.~E. Shannon, \emph{{A mathematical theory of communication}}, Bell system
  technical journal \textbf{27} (1948).

\bibitem[VZ12]{Vadhan2012}
Salil Vadhan and Colin~Jia Zheng, \emph{Characterizing pseudoentropy and
  simplifying pseudorandom generator constructions}, Proceedings of the 44th
  symposium on Theory of Computing (New York, NY, USA), STOC '12, ACM, 2012,
  pp.~817--836.

\bibitem[Yao82]{Yao1982}
Andrew~C. Yao, \emph{Theory and application of trapdoor functions}, Proceedings
  of the 23rd Annual Symposium on Foundations of Computer Science (Washington,
  DC, USA), SFCS '82, IEEE Computer Society, 1982, pp.~80--91.

\end{thebibliography}

\appendix

\section{Tightness of the Leakage Lemma}\label{app:LeakageLemma_Thightness}

\begin{lemma}\label{ZeroComputationalEntropy}
Let $X\in\{0,1\}^{n}$ be a random variable, $f:\,\{0,1\}^{m}\rightarrow \{0,1\}^{n}$ be a deterministic circuit of size $s$ and $ \epsilon < \frac{1}{12}$.
Then $\HmtrAv{f(X)|X}{\textup{det}\{0,1\},s}{\epsilon} < 3 $.
\end{lemma}
\begin{proof}
Consider the following distinguisher $D$: on the input $(y,x)$, where $x\in\{0,1\}^{m}$ and $y\in\{0,1\}^{n}$, run $f(x)$ and return $1$ iff $f(x)=y$. Then for every $x$ we get $D(f(x),x) = 1$. Let $Y$ be any random variable over $\{0,1\}^{n}$ such that $\widetilde{\mathbf{H}}_{\infty}(Y|X) \geqslant 3$. Then by Lemma \ref{Average_WorstCase_Conversion}, with probability $\frac{2}{3}$ over $x\leftarrow X$ we have
$ \mathbf{H}_{\infty}(Y|X=x) \geqslant 3-\log_{2}(3)$. Since $D(y,x) = 0$ if $y\not =x$, for any such $x$ we have
$\EvAdd{D}{Y|X=x}{y}{x} \leqslant  2^{-(3 - \log_2(3))} \leqslant \frac{3}{8}$, and thus, with probability $\frac{2}{3}$ over $x\leftarrow X$, we get $\EvAdd{D}{f(X)|X=x}{y}{x} - \EvAdd{D}{Y|X=x}{y}{x} \geqslant \frac{5}{8}$.
Taking the expectation over $x\leftarrow X$ we obtain finally $\mathbf{E}D(  f(X), X )- \mathbf{E}D ( Y,X ) \geqslant \frac{2}{3}\cdot \frac{5}{8} - \frac{1}{3} \cdot 1 = \frac{1}{12}$.
\end{proof}
\noindent We use this lemma to show that the esimate in Lemma \ref{LeakageLemma_Computational} cannot be improved:
\begin{theorem}[Tightness of the estimate in Lemma \ref{LeakageLemma_Computational}]\label{LeakageLemma_Computational_OptimalLoss}
Suppose that there exists an exponentially secure pseudorandom generator $f$. 
Then for every $m$ and $C>0$ we have $ \Hhll{f\left(U_m\right)}{\textup{rand}\{0,1\}, 2^{\bigO{m}}}{\frac{1}{2^{\bigO{m}}}} \geqslant m+C$ and simultaneously $ \HmtrAv{\left.f\left(U_m\right)\right|U_m}{\textup{det}\{0,1\}, \mathrm{poly}(m)}{\frac{1}{\mathrm{poly}(m)}} \leqslant 3$.
\end{theorem}
\begin{proof}
The first inequality follows from the definition of the exponentially secure pseudorandom generator.
The second inequality is implied by Lemma \ref{ZeroComputationalEntropy}.
\end{proof}

\section{Metric Entropy vs Different Kinds of Distinugishers}\label{app:TechnicalResults}

\noindent Below we prove the equivalence between boolean and real valued distinguishers

\begin{theorem}\label{nonaverage_passing_to_realvalued}
For any random variables $X,Z$ over $\{0,1\}^{n},\{0,1\}^{m}$ we have $
 \mathbf{H}_{}^{\textup{Metric},\textup{det}[0,1], s',\epsilon}(X|Z) = \mathbf{H}_{}^{\textup{Metric},\textup{det}\{0,1\}, s,\epsilon}(X|Z) 
$, where $s'\approx s$.
\end{theorem}
\begin{proof}
We only need to prove $\Hmtr{X|Z}{\textup{det}[0,1], s'}{\epsilon} \geqslant \mathbf{H}_{\infty}^{\textup{Metric},\textup{det}\{0,1\}, s,\epsilon}$ as the other direction is trivial (because the class $(\textup{det}[0,1],s)$ is larger than $(\textup{det}\{0,1\},s)$). Suppose that $\Hmtr{X|Z}{\textup{det}[0,1],s}{\epsilon} <k$. Then for some $D$ and all $Y$ satisfying $\Hm{X|Z}\geqslant k$ we have $ \left| \mathbf{E}_{(x,z)\leftarrow (X,Z)}D(x,z) - \mathbf{E}_{(x,z)\leftarrow (Y,Z)}D(x,z) \right| \geqslant \epsilon$.
Applying the same reasoning as in Thm. \ref{Metric_implies_Modulus} we can replace  $D$ with $D'$, which is equal either to $D$ or to $D^{c}$, obtaining for all distributions $\Hm{Y|Z}\geqslant k$, the following:
\begin{equation*}
 \mathbf{E}D'(X,Z) - \mathbf{E}D'(Y,Z) \geqslant \epsilon.
\end{equation*}
Consider the distribution $\left(Y^{+},Z\right)$ minimizing the left side of the above inequality. Equivalently, it maximizes the expected value of $D'$ under the condition $\Hm{Y|Z}\geqslant k$. Since this condition means that $\Hm{\left.Y^{+}\right|Z=z}\geqslant k$ for all $z$, we conclude that $\left.Y^{+}\right|Z=z$, for fixed $z$, is distributed over $2^k$ values of $x$ giving the greatest values of $D'(x,z)$. Calculating the expected values in the last inequality via integration of the tail yields
\begin{equation*}
 \int\limits_{t\in [0,1]}\mathbf{P}_{(x,z)\leftarrow (X,Z)}\left[D(x,z) > t\right] \mbox{d}t - \int\limits_{t\in [0,1]}\mathbf{P}_{(x,z)\leftarrow \left(Y^{+},Z\right)} \left[D(x,z) > t\right]\mbox{d}t \geqslant \epsilon
\end{equation*}
therefore for some number $t\in(0,1)$, the following holds:
\begin{equation*}
 \mathbf{P}_{(x,z)\leftarrow (X,Z)}\left[D(x,z) > t\right]  \geqslant \mathbf{P}_{(x,z)\leftarrow \left(Y^{+},Z\right)} \left[D(x,z) > t\right] + \epsilon.
\end{equation*}
Let $D''$ be a $\{0,1\}$-distinguisher that for every $(x,z)$ outputs $1$ iff $D(x,z) > t$.  Clearly $D''$ is of size $s+\mathcal{O}(1)$ and satisfies
\begin{equation*}
 \mathbf{E}_{(x,z)\leftarrow (X,Z)}D''(x,z) \geqslant \mathbf{E}_{(x,z)\leftarrow \left(Y^{+},Z\right)}D''(x,z) + \epsilon.
\end{equation*}
We assumed that $(Y,Z)$  maximizes  $\mathbf{E}D'(Y,Z)$. Now we argue that $(Y,Z)$ is also maximal for $D''$. We know that for every $z$ the distribution $Y_z$ is flat over the set $\Max_{D'(\cdot,z)}^{k}$ of $2^k$ values of $x$ corresponding to largest values of $D'(x,z)$.  It is easy to see that $\Max_{D'(\cdot,z)}^{k} = \Max_{D''(\cdot,z)}^{k}$.
Therefore, we have shown in fact that
\begin{equation*}
 \mathbf{E}_{(x,z)\leftarrow (X,Z)}D''(x,z) - \max\limits_{(Y,Z):\,\mathbf{H}_{\infty}(Y|Z)\geqslant k}\mathbf{E}_{(x,z)\leftarrow (Y,Z)}D''(x,z) \geqslant \epsilon,
\end{equation*}
which means exactly that $\Hmtr{X|Z}{\{0,1\},s'}{\epsilon} <k$.
\end{proof}

\end{document}